\documentclass[runningheads]{llncs}
\usepackage{graphicx}
\usepackage{hyperref}
\usepackage{color}
\usepackage{todonotes}
\usepackage{amsfonts,amssymb,amstext}

\newcommand{\pair}[1]{\left({#1}\right)}
\newcommand{\sqbra}[1]{\left[{#1}\right]}
\newcommand{\set}[1]{\left\{{#1}\right\}}
\newcommand{\ang}[1]{\left\langle{#1}\right\rangle}

\newcommand{\es}{\varnothing}
\newcommand{\pow}{\mathcal{P}}
\newcommand{\nat}{\mathbb{N}}

\newcommand{\real}{\mathbb{R}}

\newcommand{\IMP}{IMP}
\usepackage{tikz} 
\usetikzlibrary{graphs}
\usetikzlibrary{arrows}
\usetikzlibrary{arrows,positioning,automata,calc}

\begin{document}
\title{Applying Abstract Argumentation Theory to Cooperative Game Theory}
%
%
\author{Anthony P. Young
\orcidID{0000-0003-0747-3866} \and
David Kohan Marzag$\widetilde{\text{a}}$o
\orcidID{0000-0001-8475-7913} \and
Josh Murphy}

\authorrunning{A. P. Young, D. Kohan Marzag$\widetilde{\text{a}}$o and J. Murphy}
%
\institute{Department of Informatics, King's College London,\\ Bush House, Strand Campus, 30 Aldwych, WC2B 4BG \\ 
\{\href{mailto:peter.young@kcl.ac.uk}{peter.young}, \href{mailto:david.kohan@kcl.ac.uk}{david.kohan}, \href{mailto:josh.murphy@kcl.ac.uk}{josh.murphy}\}@kcl.ac.uk}
%
\maketitle              
%


\begin{abstract}
We apply ideas from abstract argumentation theory to study cooperative game theory. Building on Dung's results in his seminal paper, we further the correspondence between Dung's four argumentation semantics and solution concepts in cooperative game theory by showing that complete extensions (the grounded extension) correspond to Roth's subsolutions (respectively, the supercore). We then investigate the relationship between well-founded argumentation frameworks and convex games, where in each case the semantics (respectively, solution concepts) coincide; we prove that three-player convex games do not in general have well-founded argumentation frameworks.

\end{abstract}


\section{Introduction}\label{sec:intro}

\textit{Argumentation theory} is the branch of artificial intelligence (AI) that is concerned with the rational and transparent resolution of disagreements between arguments (e.g. \cite{ArgAI}). \textit{Abstract} argumentation theory, as articulated in Dung's seminal paper \cite{Dung:95}, abstracts away from the contents of the arguments and the nature of their disagreements. The resulting directed graph (digraph) representation of arguments (nodes) and their disagreements (directed edges), called an \textit{abstract argumentation framework} (AF), is simple yet powerful enough to resolve these disagreements and determine the sets of winning arguments.


Dung demonstrated the ``correctness'' of abstract argumentation by showing how abstract argumentation ``can be used to investigate the logical structure of the solutions to many practical problems'' \cite[Section 3]{Dung:95}. Specifically, he investigated two examples of problems from microeconomics (e.g. \cite{MWG:95}): cooperative game theory and matching theory. In each case, Dung showed how an appropriate AF can represent a given cooperative game or a given instance of the stable marriage problem, and that the sets of winning arguments in such AFs correspond to meaningful solutions in both of these domains.

In this paper, we further demonstrate the ``correctness'' of abstract argumentation theory by investigating its relationship with cooperative game theory. \textit{Cooperative game theory} (e.g. \cite{Chalkiadakis:11}) is the branch of game theory (e.g. \cite{vNM:44}) that studies how normatively rational agents may cooperate in order to possibly earn more payoff than they would do as individuals. Cooperative game theory abstracts away from individual agents' strategies in such games for a simpler and ``high level'' view of their interactions.\footnote{The nature of this cooperation is exogenous to the theory, but can be interpreted as groups of agents forming binding contracts. See, e.g. \cite[page 7]{Chalkiadakis:11}.}
Each \textit{cooperative game} consists of finitely many agents that can cooperate as \textit{coalitions}, and each coalition earns some payoff as measured by its \textit{value}. Given certain standard assumptions on the values of coalitions, all agents should cooperate as a single coalition, called the \textit{grand coalition}.
However, this still leaves open the question of how the payoff obtained by the grand coalition (which is already maximised) should be distributed among the individual agents, such that no agent should want to defect from the grand coalition. Historically, the first \textit{solution concept} for cooperative games that captures this idea is the \textit{Von Neumann-Morgenstern (vNM) stable set}, where each such set consists of such payoff distributions interpreted as an ``acceptable standard of behaviour'' \cite{vNM:44}.
Subsequently, Dung showed that each possible payoff distribution can be interpreted as an argument in an AF. Payoff distributions ``disagree'' when agents can defect because they can earn strictly more. This argumentative interpretation of cooperative games allowed Dung to demonstrate that the stable extensions of the AF of each cooperative game correspond exactly to the game's vNM stable sets \cite[Theorem 37]{Dung:95}.


However, just like that stable extensions of an AF may not exist, vNM stable sets for cooperative games also may not exist \cite{Lucas:67,Lucas:69}. As a result, alternative solution concepts have been proposed.
For example, Dung 
proposed that sets of payoff distributions that form preferred extensions could serve as an alternative solution concept, because preferred extensions of AFs always exist,
and therefore this is well-defined for all cooperative games. Other possible alternative solution concepts from cooperative game theory include the \textit{core} \cite{Gillies:59}, the \textit{subsolution} and the \textit{supercore} \cite{Roth:76}. Dung showed that the core corresponds to the set of unattacked arguments of the game's AF \cite[Theorem 38]{Dung:95}. This paper's first contribution is to finish the correspondence between Dung's four argumentation semantics and the various solution concepts from cooperative game theory by
proving that the complete extensions (respectively, the grounded extension) of the AF correspond(s) to the subsolutions (respectively, the supercore) of the cooperative game. 
These correspondences allow us to characterise when the supercore of a cooperative game is non-empty via the \textit{Bondareva-Shapley theorem} \cite{Bondareva:63,Shapley:67}: exactly when the game is  \textit{balanced} (see Section \ref{sec:supercore_is_grounded}).

Shapley investigated a special class of cooperative games called \textit{convex games}, which capture the intuition that agents have more incentive to join larger coalitions; the key property of each such game is that its core is its unique stable set \cite[Theorem 8]{Shapley:71}.
In abstract argumentation theory, a similar result by Dung states that if an AF is \textit{well-founded}, its grounded extension is its unique stable set \cite[Theorem 30]{Dung:95}. Given these similar results, we would like to know whether convex games always give rise to well-founded AFs. Our second contribution
in this paper is a negative answer to this question through a three-player counter-example.

The rest of this paper is structured as follows. In Section \ref{sec:background}, we review abstract argumentation theory and cooperative game theory. In Section \ref{sec:subsolution_supercore}, we complete the correspondences between Dung's four argumentation semantics with solution concepts in cooperative games and study 
the properties using well-known results from argumentation. In Section \ref{sec:convexity_well_foundedness}, we recap convex games and well-founded AFs, and give a counter-example of a three-player convex game that gives rise to a non-well-founded AF. In Section \ref{sec:related_work}, we compare our results with related work in argumentation theory and game theory, and conclude with future work.


\section{Background}\label{sec:background}

\textbf{Notation:} Let $X,Y,Z$ be sets. $\pow\pair{X}$ is the \textit{power set of $X$} and $|X|$ is the \textit{cardinality of $X$}.  $\nat$ is \textit{the set of natural numbers} (including $0$) and $\real$ is \textit{the set of real numbers}. Further, $\real^+$ (respectively, $\real^+_0$) is the \textit{set of positive (respectively, non-negative) real numbers}. For $n\in\nat$, the \textit{$n$-fold Cartesian power of $X$} is $X^n$. If $f:X\to Y$ and $g:Y\to Z$ are appropriate functions, then $g\circ f:X\to Z$ is the \textit{composition} of $g$ then $f$, and functional composition is denoted with $\circ$. If $X$ is a set, then for a function $f:X\to\real$, $f\geq 0$ abbreviates $\pair{\forall x\in X}f(x)\geq 0$.


\subsection{Abstract Argumentation Theory}\label{sec:abs_arg_recap}

An \textbf{(abstract) argumentation framework}\footnote{When defining our terms in this paper, any words in between brackets may be omitted when using the term, e.g. in this case, the terms ``argumentation framework'' and ``abstract argumentation framework'' are interchangeable.} (AF) is a directed graph (digraph) $\ang{A,R}$ where $A$ is 
the \textbf{set of arguments} and $R\subseteq A^2$ denotes 
the \textbf{attack relation}, where for $a,b\in A$, $(a,b)\in R$, abbreviated by $R(a,b)$, denotes that $a$ disagrees with $b$. Let $S\subseteq A$ be a set of arguments for the remainder of this subsection. Define $S^+:=\set{b\in A\:\vline\:\pair{\exists a\in S}R(a,b)}$ to be the set of all arguments attacked by $S$. The \textbf{neutrality function} is $n:\pow\pair{A}\to\pow\pair{A}$ where $n(S):=A-S^+$ denotes the set of arguments \textit{not} attacked by $S$, i.e. the set of arguments that $S$ is \textit{neutral} towards. We say $S$ is \textbf{conflict-free} (cf) iff $S\subseteq n(S)$. Similar to $S^+$, let $S^-:=\set{b\in A\:\vline\:\pair{\exists a\in S} R(b,a)}$, and for $a\in A$, $a^{\pm}:=\set{a}^{\pm}$. The \textbf{defence function} is $d:\pow\pair{A}\to\pow\pair{A}$ where $a\in d(S)$ iff $a^-\subseteq S^+$.\footnote{In \cite[Section 2.2]{Dung:95} this is called the \textbf{characteristic function}.} It can be shown that $d=n^2:=n\circ n$ \cite[Lemma 45]{Dung:95}. Let $U\subseteq A$ denote the \textbf{set of unattacked arguments}. It can be shown that $U=d\pair{\es}$. We say $S$ is \textbf{self-defending} (sd) iff $S\subseteq d(S)$. Further, $S$ is \textbf{admissible} iff it is cf and sd.

To determine which sets of arguments are 
justifiable we say $S$ is a \textbf{complete extension} iff $S$ is admissible and $d(S)\subseteq S$.\footnote{i.e. if one can defend a proposition then one is obliged to accept it as justifiable.} Further, $S$ is a \textbf{preferred extension} iff it is a $\subseteq$-maximal complete extension \cite[Theorem 7.11]{Young:18}, $S$ is a \textbf{stable extension} iff $n(S)=S$ and $S$ is the \textbf{grounded extension} iff it is the $\subseteq$-least complete extension. These four semantics are collectively called the \textbf{Dung semantics}, and each defines sets of winning arguments given $\ang{A,R}$.

\subsection{Cooperative Game Theory}\label{sec:recap_CGT}

We 
review the basics of cooperative game theory (see, e.g. \cite{Chalkiadakis:11}). Let $m\in\nat$ and $N=\set{1,2,3,\ldots,m}$ be our \textbf{set of players} or \textbf{agents}.\footnote{We use ``$m$'' instead of the more traditional ``$n$'' for the number of players to avoid confusion with the neutrality function $n:\pow\pair{A}\to\pow\pair{A}$ in argumentation.} We assume $m\geq 1$. A \textbf{coalition} $C$ is any subset of $N$. The \textbf{empty coalition} is $\es$ and the \textbf{grand coalition} is $N$.

\begin{example}\label{eg:coalitions}
Consider the set of players $N=\set{1,2,3}$, where $m=3$. In our examples, agents 1, 2 and 3 are respectively named Josh, David and Peter. If all three decide to work together, then they form the grand coalition $N$. If only Josh and David work together and Peter works alone, then the resulting coalitions are, respectively, $\set{1,2}$ and $\set{3}$.
\end{example}

A \textbf{valuation function} is a function $v:\pow\pair{N}\to\real$ such that $v\pair{\es}=0$. The number $v(C)$ can be thought of as the coalition's payoff as a result of the agents' coordination of strategies; this payoff is in arbitrary units.

\begin{example}\label{eg:value}
(Example \ref{eg:coalitions} continued) If Josh, David and Peter work together and earn $10$ units of payoff, then $v(N)=10$. If Josh and Peter 
will lose $50$ units of payoff if they work together, then $v(\set{1,3})=-50$.
\end{example}

\noindent Given $N$ and $v$, with $m:=|N|$, a \textbf{(cooperative) ($m$-player) game (in normal form)} is the pair $G:=\ang{N,v}$. The following five properties are standard for $v$. We say $v$ is \textbf{non-negative} iff $v\geq 0$; this excludes valuation functions such as the one in Example \ref{eg:value}. We say $v$ is \textbf{monotonic} iff for all $C,C'\subseteq N$, if $C\subseteq C'$, then $v(C)\leq v(C')$. We say $v$ is \textbf{constant-sum} iff $\pair{\forall C\subseteq N}v(C)+v(N-C)=v(N)$. We say $v$ is \textbf{super-additive} iff for all $C,C'\subseteq N$, if $C$ and $C'$ are disjoint then $v\pair{C\cup C'}\geq v\pair{C}+v\pair{C'}$. We say $v$ is \textbf{inessential} iff $\sum_{k=1}^m v\pair{\set{k}}=v\pair{N}$; inessential means there is no incentive to cooperate.

It is easy to show that if $v$ is non-negative and super-additive, then $v$ is monotonic, while the converse is not true (e.g. \cite[Example 1]{Caulier:09}). For the rest of the games in this paper, we will assume $v$ is non-negative, super-additive and \textbf{essential} (i.e. not inessential).\footnote{When combined with super-additivity, it follows that there are two disjoint coalitions $C$ and $C'$ such that $v(C\cup C')> v(C)+v(C')$.}

\begin{example}\label{eg:value2}
(Example \ref{eg:coalitions} continued) Suppose if Josh, David or Peter earn no payoff if they work as individuals, but if any two of them work together they earn $10$ units of payoff, and if all three work together they earn $20$ units of payoff. This $v$ is 
non-negative, super-additive, not constant-sum, and essential.
\end{example}


An \textbf{outcome} of a game is the pair $\pair{CS,{\bf{x}}}$, where $CS$ is a partition of $N$ called a \textbf{coalition structure}, and ${\bf{x}}\in\real^m$ is a \textbf{payoff vector} that distributes the value of each coalition to the players in that coalition. As we have assumed that $v$ is non-negative, super-additive and essential, then $v(N)$ has the (strictly) largest payoff among all coalitions by monotonicity. Agents are rational and want to maximise their payoff, and so they should seek to form the grand coalition. Therefore, we restrict our attention to outcomes where $CS=\set{N}$.

How should the amount $v(N)$ be distributed among the $m$ players? 
In this paper, we consider \textbf{transferable utility games}, which allow for the distribution of $v(N)$ arbitrarily to the $m$ players, e.g. by interpreting $v(N)$ as money, which all players should desire. This leads to the following properties of payoff vectors $\bf{x}$. We say ${\bf{x}}:=\pair{x_1,x_2,\ldots,x_m}$ is \textbf{feasible} iff $\sum_{k\in N} x_k\leq v(N)$, \textbf{efficient} iff $\sum_{k\in N} x_k= v(N)$ and \textbf{individually rational} iff $\pair{\forall k\in N}v\pair{\set{k}}\leq x_k$. We call a payoff vector $\bf{x}$ an \textbf{imputation} iff $\bf{x}$ is efficient and individually rational; intuitively, imputations distribute all the money to every agent without waste, such that every agent earns at least as much as when they work alone. Following \cite{Dung:95}, we denote the \textbf{set of imputations for a game $G$} with $\IMP(G)$, or just $\IMP$ if it is clear which cooperative game $G$ we are referring to.

\begin{example}\label{eg:imputations}
(Example \ref{eg:value2} continued) As $v(N)=20$ (which we now measure in dollars (\$), as an example of transferable utility), a feasible payoff vector is $\pair{5,5,5}$. An efficient payoff vector is $\pair{10,5,5}$. Notice that both vectors are individually rational because $\pair{\forall k\in N}v\pair{\set{k}}=0$. Therefore, $\pair{10,5,5}$ is a valid imputation, in which case Josh receives $\$10$, while David and Peter receive $\$5$ each.
\end{example}


The solution concepts of cooperative games that we will consider are concerned with whether coalitions of agents are incentivised to defect from the grand coalition.\footnote{See Section \ref{sec:related_work} for a brief mention of other solution concepts.} Given a game $G=\ang{N,v}$, let $C$ be a coalition and ${\bf{x}},{\bf{y}}\in \IMP$. We say $\bf{x}$ \textbf{dominates $\bf{y}$ via $C$}, denoted ${\bf{x}}\to_{C}{\bf{y}}$, iff (1) $\pair{\forall k\in C}x_k>y_k$ and (2) $\sum_{k\in C}x_k\leq v(C)$. Intuitively, given imputation $\bf{y}$, it is possible for a subset of players to defect from $N$ to form their own coalition $C$, where (1) they will each do strictly better because (2) they will earn enough to do so; the resulting payoff is $\bf{x}$. Note that for all $C\subseteq N$, $\to_C$ is a well-defined binary relation on $\IMP$.

\begin{example}\label{eg:domination}
(Example \ref{eg:imputations} continued) Consider the two imputations $(20,0,0)$ and $(12,4,4)$. Clearly, $(12,4,4)\to_{\set{2,3}}(20,0,0)$, because David (agent 2) and Peter (agent 3) can defect to $\set{2,3}$ and earn $4+4=8\leq v\pair{\set{2,3}}=10$,
which is strictly better than both of them getting nothing in the imputation $(20,0,0)$.
\end{example}

It is easy to see from the definition of $\to_C$ that it is irreflexive, acyclic (and hence asymmetric), and transitive. Further, some important special cases include $\to_{N}=\es$ and $\pair{\forall k\in N}\to_{\set{k}}=\es$, i.e. it does not make sense for the grand coalition to defect to the grand coalition, and individual players cannot defect from $N$ (the latter due to individual rationality). Also, $\to_{\es}=\IMP^2$, the total binary relation on $\IMP$. We say imputation ${\bf{x}}$ \textbf{dominates} imputation $\bf{y}$, denoted ${\bf{x}}\to{\bf{y}}$, iff $\pair{\exists C\subseteq N}\sqbra{C\neq\es,\:{\bf{x}}\to_{C}{\bf{y}}}$. This is a well-defined irreflexive binary relation on $\IMP$. However, it can be shown that this is not generally transitive, complete or acyclic (e.g. \cite[Chapter 4]{Thomas:12}). Therefore, each cooperative game gives rise to an associated directed graph $\ang{IMP,\to}$, called an \textbf{abstract game} \cite{Roth:76}.

Let $I\subseteq \IMP$ be a set of imputations. Following \cite{vNM:44}, we say $I$ is \textbf{internally stable} iff no imputation in $I$ dominates another in $I$. Further, $I$ is \textbf{externally stable} iff every imputation not belonging to $I$ is dominated by an imputation from $I$. A \textbf{(von Neumann-Morgenstern) stable set} is a subset of $\IMP$ that is both internally and externally stable. Taken together, a stable set of imputations contains the distributions of the amount of money $v(N)$ to the set of $m$ players that are socially acceptable \cite{vNM:44}.



We now recapitulate a simplification to coalitional games that does not lose generality (e.g. \cite[Chapter 4]{Thomas:12}). Let $G:=\ang{N,v}$ and $G':=\ang{N, v'}$ be two games on the same set of players. We say $G$ and $G'$ are \textbf{strategically equivalent} iff $\pair{\exists K\in\real^+}\pair{\exists{\bf{c}}\in\real^m}\pair{\forall C\subseteq N}v'(C)=Kv(C)+\sum_{k\in C}c_k$, for ${\bf{c}}:=\pair{c_1,\ldots,c_m}$, and we denote this with $G\cong G'$; this is an equivalence relation between games. Further, the function $f:\IMP(G)\to \IMP(G')$ with rule $f\pair{\bf{x}}:=K{\bf{x}}+{\bf{c}}$ is a digraph isomorphism from $\ang{\IMP(G),\to}$ to $\ang{\IMP(G'),\to'}$, where $\to'$ is the corresponding domination relation on $\IMP(G')$. It follows that if $S\subseteq \IMP(G)$ is a stable set of $G$, then its image set $f(S)\subseteq \IMP(G')$ is also a stable set of $G'$. By setting $\frac{1}{K}=v(N)-\sum_{k\in N}v\pair{\set{k}}$ and $\pair{\forall k\in N}c_k=-Kv\pair{\set{k}}$ for this $K$, then we can transform $G=\ang{N,v}$ to its \textbf{$(0,1)$-normalised form}, $G^{(0,1)}:=\ang{N,v^{(0,1)}}$, such that $\pair{\forall k\in N}v\pair{\set{k}}= 0$ and $v(N)=1$. 

\begin{example}\label{eg:normalised}
(Example \ref{eg:domination} continued) We have that $\bf{c}=0$ and $K=\frac{1}{20}$, therefore $v^{(0,1)}\pair{C}=0$ if $|C|\leq 1$, $v^{(0,1)}\pair{C}=\frac{1}{2}$ if $|C|=2$, and $v^{(0,1)}(N)=1$.
\end{example}

This means $\IMP\pair{G^{(0,1)}}\subseteq\real^m$ is the standard $(m-1)$-dimensional topological simplex, which has the set $\set{(x_1,x_2,\ldots,x_m)\in\pair{\real^+_0}^m\:\vline\:\sum_{k=1}^m x_k=1}$.

\begin{corollary}\label{cor:uncountably_infinite}
As a set, $\IMP\pair{G^{\pair{0,1}}}$ is uncountably infinite.
\end{corollary}
\begin{proof}
Let $\cong$ denote bijection between sets and $\hookrightarrow$ denote an injective embedding between sets, then we have $\real\cong\sqbra{0,1}\hookrightarrow \IMP\pair{G^{\pair{0,1}}}\subseteq\real^m\cong\real$, where $\hookrightarrow$ in this case is the injective function with rule $t\mapsto K\pair{1-t, t,x_3,\ldots, x_m}$, where $\frac{1}{K}:=1+\sum_{k=3}^m x_k$ normalises the output to be on the simplex. By the Cantor-Schr\"oder-Bernstein theorem \cite[Theorem 3.2]{Jech:03}, $\IMP\pair{G^{\pair{0,1}}}\cong\real$.
\end{proof}

\noindent Therefore, any abstract game $\ang{\IMP,\to}$ has uncountably infinitely many nodes. From now, we assume all games $G$ have $v$ that are non-negative, super-additive, not necessarily constant-sum, have transferable utility, and are $(0,1)$-normalised.


\subsection{From Cooperative Game Theory to Abstract Argumentation}

In this section we recap \cite[Section 3.1]{Dung:95}. We now understand how an abstract game $\ang{\IMP,\to}$, which is also an uncountably infinite digraph, arises from a game $G$. Dung interprets $\ang{IMP,\to}$ as an abstract AF, where each argument is an argument for a given payoff distribution among the agents, and each attack denotes the possibility for a subset of agents to defect from the grand coalition. Corollary \ref{cor:uncountably_infinite} states that such an AF has uncountably infinite arguments, but this is not a problem because Dung's argumentation semantics and their properties hold for AFs of arbitrary cardinalities \cite{Baumann:15,Dung:95,Young:18}. Dung then proved that various methods of resolving conflicts in $\ang{IMP,\to}$ as an AF correspond to meaningful solution concepts of $G$. The following result is straightforward to show.

\begin{theorem}
\cite[Theorem 37]{Dung:95} Let $G$ be a game and $\ang{\IMP,\to}$ be its abstract game. If we view $\ang{\IMP,\to}$ as an AF, then each of its stable extensions is a stable set of $G$, and each stable set of $G$ is a stable extension of $\ang{IMP,\to}$.
\end{theorem}

\noindent As stable sets may not exist for AFs, and in particular there are games without stable sets \cite{Lucas:67,Lucas:69}, Dung 
proposed that preferred extensions, as they always exist \cite[Theorem 11(2)]{Dung:95}, can serve as an alternative solution concept for a game $G$ in cases where stable sets do not exist, because the properties and motivations of preferred extensions also capture the imputations that are rational wealth distributions among the $m$ players \cite[Section 3.1]{Dung:95}.

Further, another important solution concept in cooperative game theory is the core \cite{Gillies:59}. Formally, the \textbf{core} of a cooperative game $G$ is the set of imputations $\bf{x}$ satisfying the system of inequalities $\pair{\forall C\subseteq N}\sum_{k\in C}x_k\geq v(C)$.
Intuitively, the core is the set of imputations where each agent is receiving at least as much payoff even if a subset of such agents were to defect to a new coalition (regardless of how the payoff is shared within that coalition). Therefore, no agent has an incentive to defect. It can be shown that the core is the subset of imputations that are not dominated by any other imputation. It follows that:

\begin{theorem}
\cite[Theorem 38]{Dung:95} Let $G$ be a game and $\ang{\IMP,\to}$ be its abstract game. If we view $\ang{\IMP,\to}$ as an AF, then its set of unattacked arguments corresponds exactly to the core.
\end{theorem}

\noindent From argumentation theory, we thus conclude the well-known result from cooperative game theory that stable sets, if they exist, always contain the core.

\begin{example}\label{eg:core}
(Example \ref{eg:normalised} continued) We have $N=\set{1,2,3}$ and $v(C)=0$ if $|C|\leq 1$ else $v(C)=\frac{1}{2}$, and $v(N)=1$. The eight possible coalitions $C\subseteq N$ give rise to the three inequalities $x_1+x_2\geq\frac{1}{2}, x_2+x_3\geq\frac{1}{2}$ and $ x_3+x_1\geq\frac{1}{2}$. Therefore, the core consists of all imputations ${\bf{x}}=\pair{x_1,x_2,x_3}$ whose components satisfy these three inequalities, for example, $\pair{\frac{1}{3},\frac{1}{3},\frac{1}{3}}$ and $\pair{\frac{1}{2},\frac{1}{2},0}$.
\end{example}

\section{Complete Extensions and the Grounded Extension}\label{sec:subsolution_supercore}

Given that stable extensions correspond to stable sets, and the set of unattacked arguments corresponds to the core, do the complete extensions (including the preferred extensions) and the grounded extension also correspond to solution concepts in cooperative games? We now show that the answer is yes.

\subsection{Complete Extensions Correspond to Subsolutions}


As preferred extensions are a subset of complete extensions, it is natural to ask whether complete extensions more generally correspond to solution concepts in cooperative games. In \cite{Roth:76}, motivated by the general lack of existence of stable sets \cite{Lucas:67,Lucas:69}, Roth considered abstract games arising from cooperative games (in his notation) $\ang{X,>}$, where $X$ is the set of the game's \textbf{outcomes} and $>\subseteq X^2$ is an abstract \textbf{domination} relation. Let $u:\pow\pair{X}\to\pow\pair{X}$ be the function $u(S):=X-S^+$,\footnote{In \cite{Roth:76}, Roth uses $U$ for this function. Here, we use $u$ to avoid confusion with the set of unattacked arguments in an AF (defined in Section \ref{sec:abs_arg_recap}).} where in this case $S^+:=\set{y\in X\:\vline\:\pair{\exists x\in S}x>y}$. Roth then defined the \textit{subsolution} of such an abstract game as follows.

\begin{definition}
\cite[Section 2]{Roth:76} Let $\ang{X,>}$ be an abstract game. A \textbf{subsolution} is a set $S\subseteq X$ such that $S\subseteq u(S)$ and $S=u^2(S):=u\circ u(S)$.
\end{definition}

Immediately we can see that by interpreting $\ang{X,>}$ as an abstract argumentation framework, subsolutions are precisely the complete extensions.

\begin{theorem}\label{thm:complete_subsolutions}
Let $G=\ang{N,v}$ be a game and $\ang{\IMP,\to}$ be its abstract game. When seen as an argumentation framework, the complete extensions $\ang{\IMP,\to}$ are precisely the subsolutions.
\end{theorem}
\begin{proof}
$S\subseteq \IMP$ is a complete extension of $\ang{\IMP,\to}$ iff $S\subseteq n(S)$ and $S=d(S)$, where $n$ is the neutrality function $n(S)=\IMP-S^+$, but as $n^2=d$ \cite[Lemma 45]{Dung:95}, this is equivalent to saying that $S$ is a subsolution, by identifying $X=\IMP$, $>\:=\:\to$ and $u=n$.
\end{proof}

Roth has shown that every abstract game, and hence every cooperative game, has a subsolution \cite[Theorem 1]{Roth:76}.\footnote{Also, see the abstract lattice-theoretic proof in \cite{Roth:75}.}
The $\subseteq$-maximal subsolutions of $\ang{\IMP,\to}$ are exactly the preferred extensions of $\ang{\IMP,\to}$ when seen as an abstract argumentation framework. Roth also showed that stable sets are subsolutions, and hence subsolutions generalise stable sets. This is well known in argumentation theory as stable extensions are also complete extensions. We can also apply further results from abstract argumentation theory to infer more properties of subsolutions. For instance, the core is contained in all subsolutions.

\begin{corollary}\label{cor:core_in_subsolution}
Every subsolution is a superset of the core.
\end{corollary}
\begin{proof}
Interpreting $\ang{\IMP,\to}$ as an argumentation framework, the core corresponds to the set of unattacked arguments, which is $d\pair{\es}$, where $d$ is the defence function (Section \ref{sec:abs_arg_recap}). As a subsolution $S$ is a complete extension, we know that $d(S)=S$. As $d$ is $\subseteq$-monotonic and $\es\subseteq S$, we have $d\pair{\es}\subseteq S$.
\end{proof}

\noindent Further, subsolutions have a specific lattice-theoretic structure:

\begin{theorem}\label{thm:complete_semilattice}
The family of subsolutions of an abstract game form a complete semilattice that is also directed-complete.
\end{theorem}
\begin{proof}
This follows from \cite[Theorem 25(3)]{Dung:95} and \cite[Theorem 6.30]{Young:18}, respectively.
\end{proof}

\noindent We will give an example of subsolutions in Section \ref{sec:convexity_well_foundedness} (Example \ref{eg:solution}).


\subsection{The Supercore is the Grounded Extension}\label{sec:supercore_is_grounded}

In \cite[Example 5.1]{Roth:76}, Roth showed that subsolutions are in general not unique for abstract games; the supercore is one ``natural'' way of selecting a unique subsolution.

\begin{definition}
\cite[Section 3]{Roth:76} The \textbf{supercore} of an abstract game is the intersection of all its subsolutions.
\end{definition}

\noindent Immediately we can conclude the following.

\begin{theorem}\label{thm:grounded_supercore}
The supercore of an abstract game is its grounded extension when viewed as an argumentation framework.
\end{theorem}
\begin{proof}
This follows from e.g. \cite[Theorem 25(5)]{Dung:95}, or \cite[Corollary 6.8]{Young:18}.
\end{proof}

\noindent It follows that the supercore exists and is unique for all abstract games, and hence cooperative games. Further, the supercore is a special case of a subsolution because the grounded extension is complete. Also, the supercore contains the core, because the grounded extension contains all unattacked arguments, by Corollary \ref{cor:core_in_subsolution}. We will give an example of the supercore in Section \ref{sec:convexity_well_foundedness} (Example \ref{eg:solution}).

Therefore, for arbitrary cooperative games, if stable sets exist, they can serve as the possible sets of recommended payoff distributions, i.e. the ``acceptable standards of behaviour'' \cite{vNM:44}. If they do not exist, then we may use subsolutions instead. If we desire a unique subsolution, we can use the supercore.

However, Roth noted that the supercore is empty iff the core is empty. This corresponds to the well-known result in argumentation theory that the set of unattacked arguments is empty iff the grounded extension is empty (e.g. \cite[Corollary 6.9]{Young:18}). We can therefore completely characterise cooperative games with non-empty supercores using the \textit{Bondareva-Shapley theorem} \cite{Bondareva:63,Shapley:67}. Let $G=\ang{N,v}$ be a cooperative game. A function $f:\pow\pair{N}\to\sqbra{0,1}$ is \textbf{balanced} iff $\pair{\forall k\in N}\sum_{S\subseteq N-\set{k}}f\pair{S\cup\set{k}}=1$, i.e. if for every player the value under $f$ of all coalitions containing that player sum to one. We call $G$ \textbf{balanced} iff for every balanced function $f$, $\sum_{S\subseteq N}f\pair{S}v(S)\leq v(N)$.
Intuitively, each player $k$ allocates a fraction of his or her time $f\pair{S\cup\set{k}}$ to the coalition $v(S)$, and that coalition receives a value proportional to that agent's time spent there.

\begin{theorem}\label{thm:BS_thm}
(Bondareva-Shapley) A game has a non-empty core iff it is balanced.
\end{theorem}

\noindent It follows that:

\begin{corollary}\label{cor:supercore_nonempty_balanced}
A game has a non-empty supercore iff it is balanced.
\end{corollary}
\begin{proof}
The core of a game is non-empty iff its supercore is non-empty, iff it is balanced, by the Bondareva-Shapley theorem (Theorem \ref{thm:BS_thm}).
\end{proof}

From both argumentation theory and abstract games in cooperative game theory, we have the following well-known containment relations between the solution concepts in cooperative games.

\begin{theorem}
Let $G$ be a cooperative game. Its stable sets are $\subseteq$-maximal subsolutions, which are subsolutions, and the supercore is also a subsolution.
\end{theorem}
\begin{proof}
Immediate, because in argumentation theory, stable extensions are preferred, which are complete. Further, the grounded extension is also complete.
\end{proof}

\subsection{Summary}

We summarise the correspondences between Dung's argumentation semantics and the solution concepts of cooperative games in the Table \ref{table:semantics_and_solutions},
including the results presented in this paper.

\begin{table}[h]
\centering
\caption{A Table Showing Concepts in Abstract Argumentation and the Corresponding Concept in Cooperative Games}
\begin{tabular}{|c|c|c|} \hline
\textbf{Abstract Argumentation} & 
\textbf{Cooperative Game} & \textbf{Reference} \\\hline
Set of arguments $A$ & Set of imputations $\IMP$ & \cite[Section 3.1]{Dung:95} \\\hline
Attack relation $R$ & Domination relation $\to$ & \cite[Section 3.1]{Dung:95} \\\hline
Argumentation Framework $\ang{A,R}$ & Abstract Game $\ang{\IMP,\to}$ & \cite[Section 3.1]{Dung:95} \\\hline
Unattacked arguments & The Core & \cite[Thm. 38]{Dung:95} \\\hline
The Grounded Extension & The Supercore & Thm. \ref{thm:grounded_supercore} \\\hline
Complete Extensions & Subsolutions & Thm. \ref{thm:complete_subsolutions} \\\hline
Preferred Extensions & $\subseteq$-maximal Subsolutions & \cite[Section 3]{Dung:95}, Thm. \ref{thm:complete_subsolutions} \\\hline
Stable Extensions & Stable Sets & \cite[Thm. 37]{Dung:95} \\\hline
\end{tabular}
\label{table:semantics_and_solutions}
\end{table}

\noindent Further, we have shown that the supercore exists, and is unique and non-empty, iff the game is balanced (Corollary \ref{cor:supercore_nonempty_balanced}). Finally, the family of subsolutions form a complete semilattice that is also directed-complete (Theorem \ref{thm:complete_semilattice}). These correspondences allow us to apply ideas from argumentation to cooperative game theory, as we will in the next section.


\section{Convex three-player Cooperative Games and Well-Founded Argumentation Frameworks}\label{sec:convexity_well_foundedness}

Having shown how abstract argumentation can be used in cooperative game theory, we now investigate the relationship between well-founded AFs and convex games, because in both cases the semantics (respectively, solution concepts) coincide.

\subsection{Convex Games and Coincidence of Solution Concepts}\label{sec:convex}

An important type of cooperative game is that of a convex game \cite{Shapley:71}. Formally, $G=\ang{N,v}$ is \textbf{convex} iff $\pair{\forall C,C'\subseteq N}v\pair{C\cup C'}+v\pair{C\cap C'}\geq v(C)+v(C')$. Clearly, convex games are super-additive. This property is equivalent to \cite[Proposition 2.8]{Chalkiadakis:11}: for all $S,T\subseteq N$, if $T\subseteq S$, then $\pair{\forall k\in N-S} v\pair{T\cup\set{k}}-v(T)\leq v\pair{S\cup\set{k}}-v(S)$. Intuitively, this means as a coalition grows in size, there is more incentive for agents not already in the coalition to join. Shapley calls this a ``band-wagon'' effect \cite{Shapley:71}. Further, if $G\cong G'$ and $G$ is convex, then $G'$ is also convex.\footnote{\label{fn:convex_strategic_equivalence} This can be shown by writing out the definition of convex for the coalitions in $G'$ given the definition of strategic equivalence in Section \ref{sec:recap_CGT}, and then applying the inclusion-exclusion principle for sets.}  The key property of convex games that we focus on is:

\begin{theorem}
\cite[Theorem 8]{Shapley:71} The core of a convex game is stable.
\end{theorem}

\noindent From our results in Section \ref{sec:subsolution_supercore}, we can immediately show the following.

\begin{corollary}\label{cor:unattacked_stable}
If $G$ is convex, then the set of unattacked arguments of its AF $\ang{\IMP,\to}$ is the only stable extension.
\end{corollary}
\begin{proof}
Immediate from \cite[Theorems 37 and 38]{Dung:95}.
\end{proof}

\noindent Convex games exhibit a coincidence of the solution concepts so far considered.

\begin{corollary}\label{cor:unattacked_stable2}
If $G$ is convex, then its core is also its supercore, which is also its unique subsolution.
\end{corollary}
\begin{proof}
If $G$ is convex, then viewing its abstract game $\ang{IMP,\to}$ as an AF, its set $U$ of unattacked arguments is the unique stable extension \cite[Theorem 37]{Dung:95}. But if $U$ is stable, then $U$ is also the grounded extension - else the grounded extension, as a superset of $U$, is not conflict-free. Thus $U$ is also the supercore of $G$ (Theorem \ref{thm:grounded_supercore}). As $U$ is unique, it is the unique subsolution of $G$.
\end{proof}

\subsection{Well-Founded Argumentation Frameworks}

Now recall from abstract argumentation theory we have a sufficient condition for an AF to have all four of Dung's argumentation semantics coincide.

\begin{theorem}\label{thm:well_founded_one_winning_set}
\cite[Theorem 30]{Dung:95} If $\ang{A,R}$ is well-founded, i.e. there is no $A$-sequence\footnote{\label{fn:X-sequence} Let $X$ be a set, an \textit{$X$-sequence} is a function $f:\nat\to X$ written as $\set{x^i}_{i\in\nat}$, so $\pair{\forall i\in\nat}x^i\in X$.} $\set{a^i}_{i\in\nat}$ such that $\pair{\forall i\in\nat}R(a^{i+1},a^{i})$ \cite[Definition 29]{Dung:95}, then its grounded extension is the unique stable, complete and preferred extension.
\end{theorem}

\noindent Given that the consequences of Theorem \ref{thm:well_founded_one_winning_set} and Corollaries \ref{cor:unattacked_stable} and \ref{cor:unattacked_stable2} are the same, one may ask whether the AFs arising from convex games is in some sense ``stronger'' than well-founded AFs. Indeed, convex games refer to unattacked arguments $U$, while well-founded AFs refer to the grounded extension, a superset of $U$. Could it be that convex games always give rise to well-founded AFs? We answer this question in the following section.





\subsection{Three-Player Convex Games}

To make this problem more tractable, we specialise to three-player convex games. We will assume the following canonical form without loss of generality:

\begin{theorem}\label{thm:standard_form_three_player}
\cite[Slide 19]{Lui:08} Every essential three-player game that is not necessarily constant-sum is strategically equivalent to the $(0,1)$-normalised three-player game $\ang{N,v^{(0,1)}}$ where $v^{(0,1)}(C)=0$ if $|C|\leq 1$, $v^{(0,1)}(N)=1$, and for $a,b,c\in\sqbra{0,1}$, $v^{(0,1)}\pair{\set{1,2}}=a$, $v^{(0,1)}\pair{\set{2,3}}=b$ and $v^{(0,1)}\pair{\set{3,1}}=c$.
\end{theorem}

\noindent Convexity constrains the parameters $a,b$ and $c$ as follows.

\begin{corollary}\label{cor:convex_inequalities}
Every essential three-player convex game $G$ that is not necessarily constant-sum is strategically equivalent to the $(0,1)$-normalised game $\ang{N,v^{(0,1)}}$ defined in Theorem \ref{thm:standard_form_three_player} iff $a+b\leq 1$, $b+c\leq 1$ and $c+a\leq 1$.
\end{corollary}
\begin{proof}
(Sketch) ($\Rightarrow$) If $G$ is convex and $G\cong \ang{N,v^{(0,1)}}$, then $\ang{N,v^{(0,1)}}$ is convex (Footnote \ref{fn:convex_strategic_equivalence}). There are $2^3=8$ distinct coalitions, which we can use to write out the inequality in the definition of convex to conclude the resulting three inequalities on $a$, $b$ and $c$. ($\Leftarrow$) If $\ang{N,v^{(0,1)}}$ is a game such that $a$, $b$ and $c$ satisfies the three inequalities, then $\ang{N,v^{(0,1)}}$ is convex (Section \ref{sec:convex}) and any game strategically equivalent to it, in particular $G$, is also convex.
\end{proof}

\begin{example}\label{eg:solution}
(Example \ref{eg:core} continued) Our game is convex, as $a=b=c=\frac{1}{2}$, by Corollary \ref{cor:convex_inequalities}. Therefore, the core of this game, calculated in Example \ref{eg:core} to be the set of imputations ${\bf{x}}=(x_1,x_2,x_3)$ such that $x_1+x_2\geq\frac{1}{2}$, $x_2+x_3\geq\frac{1}{2}$ and $x_3+x_1\geq\frac{1}{2}$, is also the supercore, the only subsolution, and the only stable set.
\end{example}

The next result shows that three-player convex games do not always give rise to well-founded AFs.

\begin{theorem}\label{thm:convex_does_not_mean_well_founded}
The game in Examples \ref{eg:coalitions} to \ref{eg:solution} is an essential, non-constant-sum three-player convex game whose AF is not well-founded.
\end{theorem}
\begin{proof}
Example \ref{eg:coalitions} states this game is three-player, Example \ref{eg:solution} states that this game is convex, and Example \ref{eg:value2} states that this game is essential and not constant-sum. Let $\ang{IMP,\to}$ denote the abstract game from our examples, now seen as an AF. Consider the following $IMP$-sequence $\set{{\bf{x}}^i}_{i\in\nat}$ where
\begin{equation}\label{eq:counter}
{\bf{x}}^i:=\pair{x^i_1, x^i_2, x^i_3}=\pair{\frac{1}{4}-\frac{1}{2(i+2)},\: \frac{1}{4}-\frac{1}{2(i+2)},\: \frac{1}{2}+\frac{1}{i+2}}.
\end{equation}
Clearly, this is a well-defined imputation for all $i\in\nat$, because the three components sum to 1 and each component is non-negative.

We now show that $\pair{\forall i\in\nat}{\bf{x}}^{i+1}\to{\bf{x}}^{i}$, and hence $\ang{IMP,\to}$ is not a well-founded AF. We only need domination with respect to the coalition $\set{1,2}$. For any $i\in\nat$, we have ${\bf{x}}^{i+1}\to_{\set{1,2}}{\bf{x}}^i$ because (1) the agents 1 and 2 do strictly better, i.e. $x^{i+1}_1>x^i_1$ and $x^{i+1}_2>x^i_2$, which in our case is
\begin{equation}
\frac{1}{4}-\frac{1}{2(i+2)} > \frac{1}{4}-\frac{1}{2(i+1)}\Leftrightarrow \frac{1}{i+2}<\frac{1}{i+1}\Leftrightarrow i+2 > i+1,
\end{equation}
which is true for all $i\in\nat$. Furthermore, (2) agents 1 and 2 earn enough payoff after defecting such that they can strictly better, because
\begin{equation}\label{eq:intermediate}
x_1^{i+1}+x_2^{i+1}=\frac{1}{4}-\frac{1}{2(i+2)} + \frac{1}{4}-\frac{1}{2(i+2)}\leq v^{(0,1)}\pair{\set{1,2}}=\frac{1}{2}.
\end{equation}
Equation \ref{eq:intermediate} is equivalent to
\begin{equation}
\frac{1}{2}-\frac{1}{i+2}\leq \frac{1}{2}\Leftrightarrow\frac{1}{i+2}\geq 0,
\end{equation}
which is true for all $i\in\nat$. Therefore, $\pair{\forall i\in\nat}{\bf{x}}^{i+1}\to{\bf{x}}^i$ for $\ang{IMP,\to}$ of the convex game where $a=b=c=\frac{1}{2}$. Therefore, the abstract game from our running example game seen as an AF is not well-founded.
\end{proof}

\noindent Therefore, not all convex games give rise to well-founded AFs.
This result clarifies that the coincidence of solution concepts due to convexity and the coincidence of argumentation semantics due to well-foundedness are of different natures.

\section{Discussion and Related Work}\label{sec:related_work}

In this paper, we have proved that for the uncountably infinite AFs that arise from cooperative games, the complete extensions (respectively the grounded extension) correspond to that game's subsolutions (respectively, supercore), by Theorem \ref{thm:complete_subsolutions} (respectively, Theorem \ref{thm:grounded_supercore}). This allows for more results from argumentation theory to be applied to cooperative game theory, for example, the lattice-theoretic structure of the complete extensions (Theorem \ref{thm:complete_semilattice}) or when the supercore is empty (Corollary \ref{cor:supercore_nonempty_balanced}). Both convex games and well-founded AFs result in a coincidence of, respectively, solution concepts and argumentation semantics, but convex games do not necessarily give rise to well-founded AFs (Theorem \ref{thm:convex_does_not_mean_well_founded}). To the best of our knowledge, these contributions are original.\footnote{The authors have checked all papers citing \cite{Roth:76}, which is the first paper defining the supercore and subsolutions from the abstract game of a cooperative game, and found no papers on argumentation theory among them.} These efforts strengthen the ``correctness'' of abstract argumentation by demonstrating its ability to reason about problems of societal or strategic concern.

Our first contribution completes the correspondence between Dung's original four argumentation semantics with solution concepts in cooperative games. We can use this correspondence in future work to investigate the relationship between further abstract argumentation semantics not mentioned in \cite{Dung:95} (e.g. those mentioned in \cite{Baumann:15}) and solution concepts in cooperative game theory. We can also investigate \textit{continuum} AFs more generally, where the set of arguments $A\cong\real$, as cooperative game theory provides a natural motivation for them.

Our second contribution can be developed further. Intuitively, convex games do not have to be well-founded due to the continuum nature of the simplex and the order-theoretic underpinnings of domination, and hence one may divide extra payoffs into smaller and smaller units. But then one might justifiably ask whether it still makes sense for the first two agents to be be sensitive to infinitesimal improvements in their payoffs when $i$ is very large, and thus still maintain their desire to defect. We can attempt to answer this in future work.


There has been much work investigating the relationship between argumentation and game theory more generally. For example, Rahwan and Larson have used argumentation theory to analyse non-cooperative games \cite{Rahwan:09} and study mechanism design \cite{Rahwan:08}. Matt and Toni, and Baroni et al. have applied von Neumann's minimax theorem \cite{vNM:44} to measure argument strength \cite{Baroni:17,Matt:08}. Riveret et al. investigate a dialogical setting of argumentation by representing the dialogue in game-theoretic terms, allowing them to determine optimal strategies for the participants \cite{riveret:08}. Roth et al. articulate a prescriptive model of strategic dialogue by using concepts from game theory \cite{roth:07}. Our paper is distinct from these as it applies ideas from argumentation theory to investigate \textit{cooperative} games.

This paper builds on results from Dung's seminal paper \cite[Section 3.1]{Dung:95}. There have been works applying ideas from cooperative game theory to non-monotonic reasoning and argumentation, specifically the \textit{Shapley value} \cite{Shapley:53}. For example, Hunter and Konieczny have used the Shapley value to measure inconsistency of a knowledge base \cite{Hunter:06}. Bonzon et al. have used the Shapley value to measure the relevance of arguments in the context of multiagent debate \cite{Bonzon:14}. The Shapley value, as a solution concept, is concerned with measuring the payoff to each agent given their \textit{marginal contribution} in each coalition, averaged over all coalitions; we do not consider it here in this paper as we are concerned with solution concepts to do with \textit{defection} rather than marginal contributions. Future work can build on the correspondences in this paper by considering which further solution concepts from cooperative games may be relevant for argumentation.

\bibliographystyle{plain}
\bibliography{AI32019}

\end{document}